\documentclass[a4paper]{amsart}

\usepackage{microtype}
\usepackage{todonotes}
\usepackage{tikz}
\usepackage{amsthm}
\newtheorem{thm}{Theorem}
\newtheorem{lemma}[thm]{Lemma}
\newtheorem{proposition}[thm]{Proposition}

\theoremstyle{definition}
\newtheorem{definition}[thm]{Definition}
\newtheorem{example}[thm]{Example}
\newtheorem{Remark}[thm]{Remark}

\definecolor{myred}{RGB}{153, 0, 51}
\definecolor{mygreen}{RGB}{0, 102, 0}
\definecolor{myblue}{RGB}{0,51, 204}
\definecolor{myorange}{RGB}{255,102, 0}

\renewcommand{\epsilon}{\varepsilon}

\newcommand{\FH}[1]{\left\langle #1 \right\rangle_f }
\DeclareMathOperator{\first}{\it first}
\DeclareMathOperator{\alp}{\it alph}
\DeclareMathOperator{\rank}{\it r}
\DeclareMathOperator{\FB}{\it B_f}

\def\s{\textbf{s}}
\def\r{\textbf{r}}
\def\a{\textbf{a}}
\def\b{\textbf{b}}
\def\c{\textbf{c}}

\def\u{\textbf{u}}
\def\p{\textbf{p}}
\def\q{\textbf{q}}
\def\vv{\textbf{v}}
\def\w{\textbf{w}}
\def\t{\textbf{t}}

\begin{document}
	
\title{The Intersection of $3$-Maximal Submonids}

\author{Giuseppa Castiglione}

\address{Dipartimento di Matematica e Informatica, Universit{\`a} di Palermo, Italy}
\email{giuseppa.castiglione@unipa.it}

\author{\v St\v ep\'an Holub}
\address{Department of Algebra, Faculty of Mathematics and Physics, Charles University, Prague, Czech Republic}
\email{holub@karlin.mff.cuni.cz}

\keywords{3-maximal monoids, intersection, free graph}

\begin{abstract}
Very little is known about the structure of the intersection of two $k$-generated monoids of words, even for $k=3$. Here we investigate the case of $k$-maximal monoids, that is, monoids whose basis of cardinality $k$ cannot be non-trivially decomposed into at most $k$ words. We characterize the intersection in the  case of two $3$-maximal monoids.
\end{abstract}

\maketitle

\section{Introduction}
 In this paper, we investigate the intersection of three-generated monoids of words in a special case when these monoids are $3$-maximal. A  monoid of words is $k$-maximal if its generating set cannot be non-trivially decomposed into at most $k$ (shorter) words. Obviously, the intersection of two finitely generated monoids of words is regular. However, already in the case of free monoids generated by two words, the structure of the intersection can be quite complex as we recall in Theorem \ref{th:karhu}, see \cite{Kar84,Hol19}. While monoids of three words have been classified (see \cite{Harju2003} for a survey), there is no classification of their intersection. It is useful to note, and we shall use this fact in the paper, that the general question about the structure of the intersection of two $k$-generated monoids is in fact a question about maximal solvable systems of equations over $2k$ unknowns, where the left hand sides and right hand sides are formed from disjoint sets of $k$ unknowns respectively. This indicates why the question is so difficult for $k=3$, where we have to deal with six unknowns.

It turns out, however, that when the condition of being $k$-maximal is added, the problem simplifies considerably. In \cite{Castiglione_2019}, a kind of defect theorem is shown for $2$-maximal monoids, see Theorem \ref{th:int_2max} below. In case of $3$-maximal monoids, studied in this paper, we encounter a situation which rather resembles the general case of two two-generated monoids. In fact, there is a close similarity to the related problem of binary equality sets.
In \cite{Ehrenfeucht1983}, it was shown that the binary equality set is either generated by at most two words, or it is of the form $(uw^*v)^*$. While it was later shown in \cite{Holub2003} that the latter possibility never takes place for binary equality words, we show in this paper that the set of possibilities given in the previous sentence is the exact description of intersection of two $3$-maximal monoids. This setting therefore fits, from the point of its complexity, somewhere between binary equality words, and the intersection of free two-generated monoids.

\section{Preliminaries}
Let $\Sigma^*$ ($\Sigma^+=\Sigma^*\setminus\{\epsilon\}$ resp.) be the {\em free monoid} ({\em free semigroup} resp.) freely generated by a countable set $\Sigma$ which will be fixed throughout the paper. As usually, we shall call the set $\Sigma$ an \emph{alphabet}, and understand elements of $\Sigma^*$ (resp. $\Sigma^+$) as finite words (finite nonempty words resp.) over $\Sigma$ with the monoid operation of concatenation. Note however, that $\Sigma$, understood as the set of generators satisfying $\Sigma \subseteq \Sigma^*$, is the set of words of length one, rather than a set of letters.  

We say that a word $u$ is a {\em prefix} (res. {\em proper prefix}) of $w$ and we write $u\leq w$ (resp. $u < w$), if $w=uz$ for some $z \in \Sigma^*$ (resp. $z\in \Sigma^+$).  We say that $u$ is a suffix of $w$ if $w=zu$ for some $z \in \Sigma^*$. Two words $v$ and $w$ are {\em prefix comparable} iff either $v \leq w$ or $w \leq v$. A word $w$ is \emph{primitive} if $w=v^n$ implies $n=1$ and $w=v$, otherwise it is called \emph{a power}. 
If we consider pairs of words, we say that $(u, v) \in \Sigma^* \times \Sigma^*$ is a prefix (resp. proper prefix) ) of $(r,s) \in \Sigma^* \times \Sigma^*$, and we write $(u,v) \leq (r,s)$ (resp. $(u,v) < (r,s)$), if $u \leq r$ and $v \leq s$ (resp.  $u < r$ and $v < s$).

Given $u,v \in \Sigma^*$, by $u \wedge v$ we denote the longest common prefix of $u$ and $v$. Let $u\in \Sigma^*$, by $\first(u)$ we denote the first letter of $u$.

Given a subset $X$ of $\Sigma^*$, by $X^*$ we denote the submonoid of $\Sigma^*$ generated by $X$. 
Conversely, given a submonoid $M$ of $\Sigma^*$, there exists a unique minimal (w.r.t. the set inclusion) generating set $B(M)$ of $M$, called the \emph{basis} of $M$, namely     
\begin{equation}\label{def:mgs}
B(M) = (M \setminus \{\epsilon\}) \setminus (M \setminus \{\epsilon\})^2.
\end{equation}
That is, the basis of $M$ is the set of all nonempty words of $M$ that cannot be written as a concatenation of two nonempty words of $M$. 
For an arbitrary set $X \subseteq \Sigma^*$, we shall write $B(X)$ instead of $B(X^*)$. The cardinality of $B(X)$ is the \emph{rank} of $X$, denoted $\rank(X)$.

A submonoid $M$ of $\Sigma^*$ with the basis $B$ is said to be {\em free} if any word of $M$ can be {\em uniquely} expressed as a product of elements of $B$. The basis of a free monoid is called a {\em code}.

It is well-known (see ~\cite{Tilson}) that for any set $X \subseteq \Sigma^*$ there exists the smallest free submonoid $\FH X$ of $\Sigma^*$ containing $X$. It is called the {\em free hull} of $X$. The basis of $\FH X$ is called the {\em free basis} of $X$, denoted by $\FB(X)$.
The cardinality of $\FB(X)$ is called the {\em free rank} of $X$ and denoted by $r_f(X)$.

For $w \in \FH X$, let $\first_X(w)=b_1$ where $w=b_1b_2 \cdots b_n$, $b_i \in \FB(X)$, be the unique factorization of $w$ into elements of $\FB(X)$. The words $b_1, b_1b_2,\dots ,b_1b_2 \cdots b_n$ are called {\em $X$-prefixes} of $w$. We write $u <_Z w$ if $u$ is a $X$-prefix of $w$. Moreover, given $u,w \in \FH X$, by $u \wedge_X w$ we denote the {\em longest common $X$-prefix of $u$ and $w$}.

\begin{example}\label{ex:first}
  Let $X=\{abcac, bab, ab, cacabcacb, ca\}$. The free basis is $B=\{ab, b, ca, cac\}$, hence $r_f(X)=4$. For $u= ab \cdot ca\cdot ca\cdot b\cdot cac \cdot b$ and  $w=ab\cdot cac\cdot ca\cdot ca \cdot ab \in X^*$, we have $\first_X(u)=ab$, $u \wedge w= abcac$ and $u \wedge_X w=ab$. 
\end{example}{}

We have the following well-known lemma.

\begin{lemma}\label{lm:first}
	Let $X$ a finite set of $\Sigma^*$ and $B$ its free basis. Then for each $y \in B$ there exists $u \in X$ such that $\first_X(u)=y$.
\end{lemma}

In order to see the importance of the above lemma, 
let us define the {\em free graph} of a finite set $X\subset \Sigma^+$ as the undirected graph $G(X)=(X,E_X)$ without loops where $E_X=\{[u,w] \in X \times X \mid u \neq w\  \mbox{and}\  \  \first_X(u)=\first_X(w)\}$.

Let $c(X)$ be the number of connected components of $G(X)$. By Lemma \ref{lm:first}, we now have that
$$r_f(X)=c(X),$$
which immediately implies the Defect Theorem claiming that $r_f(X) < \vert X \vert$ if $X$ is not a code (cf. \cite{Lothaire1} and \cite{BPPR79}).

\begin{example}
	 Consider $X$ of the Example \ref{ex:first}. The free graph $G(X)$ has a unique edge $[abcc, ab]$ connecting the only two words starting with $ab \in B$. Note that there is no edge between $ca$ and $cacabcacb$, since $\first_X(ca) = ca \neq cac = \first_X(cacabcacb)$.
	\end{example}

The free graph is a frequently used tool in the paper because it allows us to easily establish the free rank of a set by considering the properties of the edges of the associated free graph.

\section{$k$-maximal Monoids}

In this section we study $k$-maximal submonioids introduced in \cite{Castiglione_2019}. With $\mathcal{M}_k$ we denote the family of submonoids $M$ of $\Sigma^*$ of rank at most $k$.

\begin{definition} (cf. \cite{Castiglione_2019} )
	A submonoid $M \in \mathcal{M}_k$ is \emph{$k$-maximal} if for every $M^\prime \in \mathcal{M}_k$,  $M \subseteq M^\prime$ implies $M=M^\prime$.
\end{definition}

In other words, the elements of the basis of $M$ cannot be nontrivially factored into at most $k$ words.

\begin{example}\label{ex:ch3first}
	For every word $v\in \Sigma^+$, the submonoid $\{v\}^*$ (denoted simply by $v^*$) is $1$-maximal if and only if $v$ is a primitive word. 
	
	The submonoid $\{a,cbd,dbd\}^*$ is $3$-maximal, whereas $\{a,cbd, dcbd\}^*$ is not $3$-maximal since it is contained in $\{a,cb,d\}^*.$
	\end{example}

Let $|X|=\left|\alp(X)\right|=k$, where $\alp(X)$ is the subset of letters of $\Sigma$ occurring in the words of $X$. Then $X^*$ is $k$-maximal if and only if $X^*=\alp(X)^*$. Also, a $k$-maximal submonoid is obviously generated by primitive words. On the other hand, a finite set of $k$ primitive words does not necessarily generate a $k$-maximal submonoid of $\Sigma^*$ as we can see in Example \ref{ex:ch3first}. 
\medskip	

In \cite{Castiglione_2019}, it is proved that the basis of a $k$-maximal submonoid of $\Sigma^*$ is a bifix code, in particular its free rank is $k$.  We repeat the proof here.
\begin{proposition}\label{prop:bifix}
Let $X$ be the basis of a  $k$-maximal submonoid. Then, $X$ is a bifix code.
\end{proposition}
\begin{proof}
 If $uv, u \in X$ then $X^* \subseteq Y^*$ where $Y$ is obtained from $X$ by replacing $uv$ with $v$. Hence $X^*$ is not $k$-maximal, since $v\notin X$. Similarly for suffixes. 
\end{proof}{}

The inverse is not true, see again Example \ref{ex:ch3first} where $\{a,cbd, dcbd\}$ is a bifix code.

Submonoids generated by two words, i.e., elements of $\mathcal{M}_2$, have been extensively studied in the literature (cf.~\cite{LenSchut,Kar84,Neraud1993,lerest}) and play an important role in many fundamental aspects of combinatorics on words.
It is known (see \cite{Kar84} and \cite{Hol19}) that if $X$ and $U$ have free rank $2$, then the intersection $X^* \cap U^*$ is a free monoid generated either by at most two words or by an infinite set of words. More formally, we have the following theorem.

\begin{thm}\label{th:karhu}
Let $X=\{x,y\}$ and $U=\{u,v\}$ be two sets of $\Sigma^*$ with free rank 2, then $X^* \cap U^*$ is one of the forms 
\begin{itemize}
    \item $X^* \cap U^*=\{\gamma, \beta\}^*$, for some $\gamma, \beta \in \Sigma^*$;
    \item $X^* \cap U^*=(\beta_0+\beta(\gamma(1+\delta+\cdots+\delta^t))^*\tau)^*$, for some $\beta_0, \beta, \gamma, \delta, \tau \in \Sigma^*$ and some $t \in N.$
\end{itemize}
     \end{thm}

\begin{example}\label{ex:int}
	Let $X_1=\{abca,bc\}$ and $U_1=\{a, bcabc\}$. One can verify that $X_1^* \cap U_1^*=\{abcabc, bcabca\}^*.$ Let $X_2=\{aab, aba\}$ and $U_2=\{a, baaba\}$, then $X_2^* \cap U_2^*=(a(abaaba)^*baaba)^*.$ Note that the submonoids here considered are not $2$-maximal. Indeed, $X_1^*, U_1^* \subseteq \{a, bc\}^*$ and $X_2^*, U_2^* \subseteq \{a, b\}^*$.   
\end{example}

To our knowledge nothing is proved in general for the intersection of two monoids of free rank $3$. In \cite{Kar85}, some properties of codes with three elements are studied.

Let us turn our attention to the intersection of $k$-maximal submonoids. For the intersection of $1$-maximals, that is, for the submonoids in $\mathcal{M}_1$, we have the following important property: If $x^*$ and $u^*$ are $1$-maximal submonoids (i.e., $x$ and $u$ are primitive words) then $x^* \cap u^* = \{\epsilon\}.$ 
A generalization of this result, given in \cite{Castiglione_2019}, to the case of $2$-maximal submonoids is the following.

\begin{thm}\label{th:int_2max}
Let $X=\{x,y\}$ and $U=\{u,v\}$, with $X \neq U$, be such that $X^*$ and $U^*$ are $2$-maximal submonoids of $\Sigma^*$. If $X^* \cap U^* \neq \{\epsilon\}$, then there exists a unique primitive word $z\in \Sigma^+$ such that $X^*\cap U^*=z^*$.
\end{thm}

\begin{example}\label{ex:int2max}
Let $X^*=\{abcab,cb\}$ and $U^*=\{abc,bcb\}^*$ be two $2$-maximal submonoids of $\Sigma^*$, then their intersection is  $\{abcabcbcb\}^*$. 
\end{example}

The following example shows that Theorem \ref{th:int_2max} can not be generalized to any $k>2$.

\begin{example}
	For $k=4$  let $X=\{a, b, cd, ce \}$ and $U=\{ac, bc, da, ea\}$. It is easy to see that $X^*$ and $U^*$ are $4$-maximal and $X^* \cap U^*= \{acda, acea, bcda, bcea\}^*$. For $k=5$ the intersection can be generated by $6$ elements, see for example $X=\{a, b, cd, ce, cf\}$ and $U=\{ac, bc, da, ea, fa\}$ are two $5$-maximal submonoids and $X^* \cap U^*=\{acda, acea, acfa, bcda, bcea, bcfa\}^*$. Similar examples are easily found for $k>5$.
\end{example}

In the following section, we characterize the intersection of two $3$-maximal submonoids.

\section{The Intersection of Two $3$-maximal Submonoids}

In what follows,  $X=\{x, y, z\}$ and $U=\{u, v, w\}$ will be two \emph{distinct} three-element subsets of $\Sigma^+$ such that $X^*$ and $U^*$ are $3$-maximal.
Let also $Z=X \cup U$.

We have the following lemma.

\begin{lemma}\label{lm:freegraphZ}
The free rank of $Z$, that is, the number of connected components of $G(Z)$, is more than three. Formally,  $3 < r_f(Z) = c(Z)$.
\end{lemma}
\begin{proof}
	If $r_f(Z) \leq 3$ then the inclusions $U^*\subseteq Z^*$ and $X^*\subseteq Z^*$ imply that $X^*$ and $U^*$ are not $3$-maximal unless $X^*=U^*=Z^*$ which is excluded by the hypothesis that $X$ and $U$ are distinct. 
\end{proof}

When we search for the elements of $X^* \cap U^*$ we are searching for those words that can be decomposed both into words $x$, $y$ and $z$, and into words $u$, $v$, $w$. Consider, as  an example, the sets $X=\{abbc, da, db\}$ and $U=\{abbca, b, cdad\}$. Then $abbcabbcdadb$ is such a word as can be seen from its factorizations $abbc\cdot abbc\cdot da\cdot db=abbca\cdot b\cdot b\cdot cdad\cdot b$. 

Double factorizations of this kind are best dealt with using two ternary morphisms as follows. 
We set $A= \{\a, \b, \c \}$  and define morphisms $g,h: A^* \rightarrow \Sigma^*$ by 
$$\begin{array}{ll}
g(\a)=x   & h(\a)=u \\
g(\b)=y    & h(\b)=v \\
g(\c)=z   & h(\c)=w. \\
\end{array}
$$
For better readability, we use the boldface style for elements of $A^*$. The example above is then captured by the equality $g(\a\a\b\c)=h(\a\b\b\c\b)=abbcabbcdadb$. That is, the word $abbcabbcdadb$ has the structure $\a\a\b\c$ if considered in $X^*$ and $\a\b\b\c\b$ if considered in $U^*$. 

We say that a morphism $g: A^* \rightarrow \Sigma^*$ is {\em marked} if for each pair of letters $\a_1 \neq \a_2 \in A$ we have $\first(g(\a_1)) \neq \first(g(\a_2)).$
Furthermore, if $X$ is a finite set of $\Sigma^*$ and $B$ its free basis we say that the morphism $g$ is {\em $X$-marked} if for each pair of letters $\a_1 \neq \a_2 \in A$ we have $\first_X(g(\a_1)) \neq \first_X(g(\a_2)).$ 

\medskip

Let $g,h: A^* \rightarrow \Sigma^*$ be two morphisms. The {\em coincidence set} of $g$ and $h$ is the set defined as follows
{$$C(g,h)=\{(\r, \s) \in A^+ \times A^+ \vert \,\,g(\r)=h(\s)\}.$$}

The pairs of the coincidence set are called {\em solutions}. 
A solution is \emph{minimal} if it cannot be written as the concatenation of other solutions. That is, if $(\u, \vv) < (\r, \s)$, then $(\u, \vv)$ is not a solution. Clearly, $C(g,h)$ is freely generated by the set of minimal solutions.

The property of $k$-maximality guarantees (by Proposition \ref{prop:bifix}) 
the following lemmas that are responsible for a relatively simple structure of the intersection. In particular, complications related to the second case of Theorem \ref{th:karhu} are avoided.

\begin{lemma}\label{lm:incmorphism}
  $h(\u) \leq h(\u')$ iff $\u \leq \u'$.
\end{lemma}
\begin{proof}
If $\u < \u'$ trivially $h(\u) \leq h(\u')$. Viceversa, let $h(\u) \leq h(\u')$, if there exist $\a \neq \a' \in A$ such that $\u=\p\a\u_1$ and $\u'= \p\a'\u_1'$. Then $h(\a\u_1) < h(\a'\u_1')$ which implies that $h(\a)$ and $h(\a')$ are prefix comparable, a contradiction with Proposition \ref{prop:bifix}.
\end{proof}

\begin{lemma}\label{lm:notcompa}
	Let $(\r, \s)$ and $(\r', \s')$ be two distinct minimal solutions. Then $\r$ and $\r'$ are not prefix comparable, and $\s$ and $\s'$ are not prefix comparable.
\end{lemma}
\begin{proof}
	Assume that $\r$  and $\r'$ are prefix comparable, and assume, without loss of generality, that  $\r'=\r\q$, with $\q \in A^+$. Then $g(\r')=g(\r)g(\q)=h(\s)g(\q)=h(\s')$ and $h(\s)<h(\s')$. It follows by Lemma \ref{lm:incmorphism} that $\s < \s'$, hence $(\r', \s')$ is not minimal. Similarly, we prove that $\s$ and $\s'$ are not prefix comparable. 
\end{proof}

\begin{figure}[ht]
\centering
\begin{tikzpicture}[anchor=base, baseline, inner sep = 6pt] 
\foreach \x/\y in {1/a,2/b,3/c,4/d,5/c,6/d,7/a,8/b}
\node (dol\x) at (\x*6mm,0) {$\y$};
\foreach \x/\y in {1/a,2/b,3/c,4/d,5/c,6/d,7/a,8/b}
\node (hor\x) at (\x*6mm,.6) {$\y$};
\draw (hor1.south west) -- (hor8.south east);
\draw (dol1.south west) -- (dol8.south east);
\draw (hor1.north west) -- (dol1.south west); 
\draw (hor1.north west-|hor8.south east) -- (dol8.south east); 
\draw (hor1.north west) -- (hor8.south east|-hor1.north west); 
\draw (hor2.south east|-hor1.north west) -- (hor2.south east); 
\draw (hor4.south east|-hor1.north west) -- (hor4.south east); 
\draw (hor6.south east|-hor1.north west) -- (hor6.south east); 
\draw (hor3.south east) -- (dol3.south east); 
\draw (hor5.south east) -- (dol5.south east); 
\node at (12mm,-.55) {$\a$};
\node at (27mm,-.55) {$\c$};
\node at (42mm,-.55) {$\b$};
\node at (9mm,1.1) {$\a$};
\node at (21mm,1.1) {$\c$};
\node at (33mm,1.1) {$\c$};
\node at (45mm,1.1) {$\a$};
\end{tikzpicture}
\quad\quad
	\def\Sab{\textcolor{mygreen}{ab}}	
\def\Sc{\textcolor{myred}{c}} 
\def\Scb{\textcolor{myorange}{cb}} 
\def\Sd{\textcolor{myblue}{d}} 
\raisebox{-1em}{
\begin{tikzpicture}
\node (1a) at (0,0) {$\Sab\Sc$};
\node (2a) at (2,0) {$\Sd\Sab$};
\node (3a) at (4,0) {$\Sd\Sc$};
\node (1b) at (0,1) {$\Sab$};
\node (2b) at (2,1) {$\Scb$};
\node (3b) at (4,1) {$\Sc\Sd$};
\draw[thick] (1a)-- (1b);
\draw[thick] (2a)-- (3a);
\end{tikzpicture}
}
    \caption{A representation of the solution $(\a\c\c\a, \a\c\b)$ and the free graph of morphisms of Example \ref{ex:infinit}.}.\label{fig:infinit}
\end{figure}
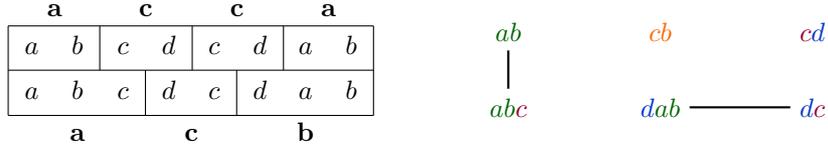

\begin{example}\label{ex:infinit}
Let 

$$\begin{array}{ll}
g(\a)=ab   & h(\a)=abc \\
g(\b)=cb    & h(\b)=dab \\
g(\c)=cd   & h(\c)=dc. \\
\end{array}$$
The pair $(\a\c\c\a, \a\c\b)$ is a solution. Indeed $g(\a\c\c\a)=abcdcdab=h(\a\c\b)$. See Figure \ref{fig:infinit} for a representation of the solution and the free graph. The free basis of $Z$ is $B=\{ab, c, cb, d\}$ and we highlight the decomposition into the free basis of $Z$ by different colors. The edges of $G_Z$ are $E_z=\{[g(\a), h(\a)],[h(\b), h(\c)]\}$. One can verify that the set of minimal solutions is $\{(\a\c^i\b,\a\c^{i+1}\a) \mid i \geq 0\}$ and the intersecton is therefore $(abc(dc)^*dab)^*$. \end{example}

This way, the problem of finding the intersection $X^* \cap U^*$ is reduced to the problem of finding minimal elements of the coincidence set of morphisms $g$ and $h$.
 Indeed, when we find a minimal solution $(\r, \s)$, with $\r,\s  \in A^*$, then $g(\r)$ (which is equal to $h(\s)$) is an element of the minimal generating set of the intersection $X^* \cap U^*$. 
 As we have seen in Example \ref{ex:infinit}, the
 intersection of two $3$-maximal submonoids can be infinitely generated. We shall see that in the case of a finite number of generators, the cardinality is at most two. A trivial example of a two generated intersection is $\{a,b,c\}^* \cap \{a,b,d\}^* = \{a,b\}^*$. A less trivial example is the following.

\begin{example}\label{ex:triangle}
	Let 
$$\begin{array}{ll}
g(\a)=ab   & h(\a)=abbc \\
g(\b)=bcdd    & h(\b)=abcb \\
g(\c)=cbdd   & h(\c)=ddab\,.
\end{array}$$
There are only two minimal solutions, namely $(\a\b\a,\a\c)$ and $(\a\c\a,\b\c)$, hence the submonoid intersection is finitely generated by $\{abbcddab, abcbddab\}$. The free basis of $Z$ is $B=\{ab,ac,bd,cd,da\}$, see Figure \ref{fig:triangle} for a representations of the two solutions and the free graph.
\end{example}

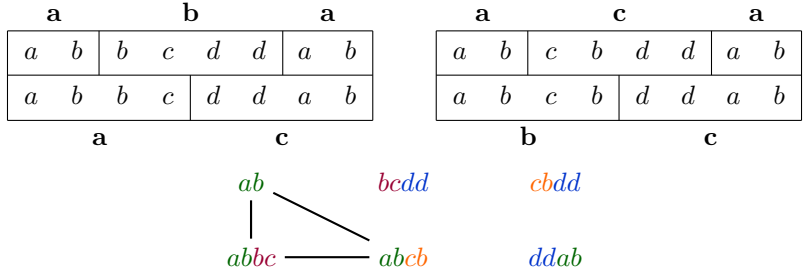
\begin{figure}[ht]
\begin{tikzpicture}[anchor=base, baseline, inner sep = 6pt] 
\foreach \x/\y in {1/a,2/b,3/b,4/c,5/d,6/d,7/a,8/b}
\node (dol\x) at (\x*6mm,0) {$\y$};
\foreach \x/\y in {1/a,2/b,3/b,4/c,5/d,6/d,7/a,8/b}
\node (hor\x) at (\x*6mm,.6) {$\y$};
\draw (hor1.south west) -- (hor8.south east);
\draw (dol1.south west) -- (dol8.south east);
\draw (hor1.north west) -- (dol1.south west); 
\draw (hor1.north west-|hor8.south east) -- (dol8.south east); 
\draw (hor1.north west) -- (hor8.south east|-hor1.north west); 
\draw (hor2.south east|-hor1.north west) -- (hor2.south east); 
\draw (hor6.south east|-hor1.north west) -- (hor6.south east); 
\draw (hor4.south east) -- (dol4.south east); 
\node at (15mm,-.55) {$\a$};
\node at (39mm,-.55) {$\c$};
\node at (9mm,1.1) {$\a$};
\node at (27mm,1.1) {$\b$};
\node at (45mm,1.1) {$\a$};
\end{tikzpicture}
\quad\quad
\begin{tikzpicture}[anchor=base, baseline, inner sep = 6pt] 
\foreach \x/\y in {1/a,2/b,3/c,4/b,5/d,6/d,7/a,8/b}
\node (dol\x) at (\x*6mm,0) {$\y$};
\foreach \x/\y in {1/a,2/b,3/c,4/b,5/d,6/d,7/a,8/b}
\node (hor\x) at (\x*6mm,.6) {$\y$};
\draw (hor1.south west) -- (hor8.south east);
\draw (dol1.south west) -- (dol8.south east);
\draw (hor1.north west) -- (dol1.south west); 
\draw (hor1.north west-|hor8.south east) -- (dol8.south east); 
\draw (hor1.north west) -- (hor8.south east|-hor1.north west); 
\draw (hor2.south east|-hor1.north west) -- (hor2.south east); 
\draw (hor6.south east|-hor1.north west) -- (hor6.south east); 
\draw (hor4.south east) -- (dol4.south east); 
\node at (15mm,-.55) {$\b$};
\node at (39mm,-.55) {$\c$};
\node at (9mm,1.1) {$\a$};
\node at (27mm,1.1) {$\c$};
\node at (45mm,1.1) {$\a$};
\end{tikzpicture}
	\def\ab{\textcolor{mygreen}{ab}}	
	\def\bc{\textcolor{myred}{bc}} 
	\def\cb{\textcolor{myorange}{cb}} 
	\def\dd{\textcolor{myblue}{dd}} 
	\centering
	\begin{tikzpicture}[]
	\node (1a) at (0,1) {$\ab$};
	\node (2a) at (2,1) {$\bc\dd$};
	\node (3a) at (4,1) {$\cb\dd$};
	\node (1b) at (0,0) {$\ab\bc$};
	\node (2b) at (2,0) {$\ab\cb$};
	\node (3b) at (4,0) {$\dd\ab$};
	\draw[thick] (1a)-- (1b);
	\draw[thick] (1a)-- (2b);
	\draw[thick] (1b)-- (2b);
	\end{tikzpicture}
	\caption{A representation of the solutions $(\a\b\a,\a\c)$ and $(\a\c\a,\b\c)$ and the free graph of morphisms of Example \ref{ex:triangle}.} \label{fig:triangle}
\end{figure}{}

In what follows, we equivalently refer to $X$ (resp. $U$) and $g(A)$ (resp. $h(A))$. Since $g(\a)$, $g(\b)$, $g(\c) \in \FH Z$ and $h(\a)$, $h(\b)$, $h(\c) \in \FH Z$ by definition, we have that $w \in \FH Z$ for any element $w$ of the intersection. 

As mentioned before, we often use the free graph of $G_Z$ as the source of information about the free basis of $Z$. The set $V_Z$ of nodes is the union of the images  $g(A)$ and $h(A)$. In figures, we graphically arrange nodes in $V_Z$ in two rows containing elements from $g(A)$ and $h(A)$ respectively.  We know that the number of connected components is the free rank of $Z$, which is at least four. Moreover, we naturally distinguish two different kinds of edges. The edges that involve nodes in the same set, either $g(A)$ or $h(A)$, are {\em horizontal edges}, and the edges that involve one node of $g(A)$ and one of $h(A)$ are {\em vertical edges}. 

The following two observations are immediate: 
\begin{itemize}
	\item A morphism $g$ is $Z$-marked iff there are no horizontal edges in the corresponding row. Indeed, by definition, $[g(\a_1),g(\a_2)] \in E_Z$ iff $\first_Z(g(\a_1))=\first_Z(g(\a_2))$. Analogously for the morphism $h$. 
	\item A solution creates a vertical edge. Indeed, if $(\r,\s) \in C(g,h)$ we have $\first_Z(g(\r_1)) = \first_Z (h(\s_1))$ and $[g(\r_1),g(\s_1)] \in E_Z$, where $\r_1 = \first (\r)$ and $\s_1 = \first (\s)$.
\end{itemize}

This implies the following property of our morphisms.

\begin{lemma}\label{lm:marked}
If $C(g,h) \neq \emptyset$ then either $g$ or $h$ is Z-marked. Moreover, if $h$ is not marked then there exist exactly two letters $\a_1,\a_2 \in A$ such that $\first_Z(h(\a_1))=\first_Z(h(\a_2))$.
\end{lemma}
\begin{proof}

\begin{figure}[ht]
	\centering
	\begin{tikzpicture}[dot/.style={circle,fill=black,
		inner sep=1.5pt},scale=.5]
	\fill[gray!30, rounded corners = 5pt] (-.3,-.3) rectangle (2.3,1.3);
	\node[dot] (1a) at (0,1) {};
	\node[dot] (2a) at (1,1) {};
	\node[dot] (3a) at (2,1) {};
	\node[dot] (1b) at (0,0) {};
	\node[dot] (2b) at (1,0) {};
	\node[dot] (3b) at (2,0) {};
	\draw[thick] (1a)-- (1b);
	\draw[thick] (1a)-- (2a);
	\draw[thick] (1b)-- (2b);
	\end{tikzpicture}	
\quad	
	\begin{tikzpicture}[dot/.style={circle,fill=black,
		inner sep=1.5pt},scale=.5]
	\fill[gray!30, rounded corners = 5pt] (-.3,-.3) rectangle (2.3,1.3);
	\node[dot] (1a) at (0,1) {};
	\node[dot] (2a) at (1,1) {};
	\node[dot] (3a) at (2,1) {};
	\node[dot] (1b) at (0,0) {};
	\node[dot] (2b) at (1,0) {};
	\node[dot] (3b) at (2,0) {};
	\draw[thick] (1a)-- (1b);
	\draw[thick] (1a)-- (2a);
	\draw[thick] (2a)-- (3a);
	\end{tikzpicture}
\quad	
\begin{tikzpicture}[dot/.style={circle,fill=black,
	inner sep=1.5pt},scale=.5]
	\fill[gray!30, rounded corners = 5pt] (-.3,-.3) rectangle (2.3,1.3);
	\node[dot] (1a) at (0,1) {};
	\node[dot] (2a) at (1,1) {};
	\node[dot] (3a) at (2,1) {};
	\node[dot] (1b) at (0,0) {};
	\node[dot] (2b) at (1,0) {};
	\node[dot] (3b) at (2,0) {};
	\draw[thick] (1a)-- (1b);
	\draw[thick] (1b)-- (2b);
	\draw[thick] (2b)-- (3b);
\end{tikzpicture}
\quad	
\begin{tikzpicture}[dot/.style={circle,fill=black,
	inner sep=1.5pt},scale=.5]
	\fill[gray!30, rounded corners = 5pt] (-.3,-.3) rectangle (2.3,1.3);
	\node[dot] (1a) at (0,1) {};
	\node[dot] (2a) at (1,1) {};
	\node[dot] (3a) at (2,1) {};
	\node[dot] (1b) at (0,0) {};
	\node[dot] (2b) at (1,0) {};
	\node[dot] (3b) at (2,0) {};
	\draw[thick] (1a)-- (1b);
	\draw[thick] (1a)-- (2a);
	\draw[thick] (2b)-- (3b);
\end{tikzpicture}
\quad	
\begin{tikzpicture}[dot/.style={circle,fill=black,
	inner sep=1.5pt},scale=.5]
\fill[gray!30, rounded corners = 5pt] (-.3,-.3) rectangle (2.3,1.3);
\node[dot] (1a) at (0,1) {};
\node[dot] (2a) at (1,1) {};
\node[dot] (3a) at (2,1) {};
\node[dot] (1b) at (0,0) {};
\node[dot] (2b) at (1,0) {};
\node[dot] (3b) at (2,0) {};
\draw[thick] (1a)-- (1b);
\draw[thick] (1b)-- (2b);
\draw[thick] (2a)-- (3a);
\end{tikzpicture}
\quad	
\begin{tikzpicture}[dot/.style={circle,fill=black,
	inner sep=1.5pt},scale=.5]
	\fill[gray!30, rounded corners = 5pt] (-.3,-.3) rectangle (2.3,1.3);
	\node[dot] (1a) at (0,1) {};
	\node[dot] (2a) at (1,1) {};
	\node[dot] (3a) at (2,1) {};
	\node[dot] (1b) at (0,0) {};
	\node[dot] (2b) at (1,0) {};
	\node[dot] (3b) at (2,0) {};
	\draw[thick] (1a)-- (1b);
	\draw[thick] (2a)-- (3a);
	\draw[thick] (2b)-- (3b);
\end{tikzpicture}
	\caption{Free graphs with a nonempty coincidence set and two horizontal edges.}
	\label{fig:components}
\end{figure}
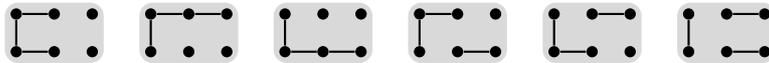

Since the set of solutions is nonempty, there is at least one vertical edge in the free graph of $Z$.
Since the free rank of $Z$ is at least four, the free graph cannot contain two horizontal edges (see Figure \ref{fig:components}). The claim follows. 
\end{proof}

Examples \ref{ex:infinit} and \ref{ex:triangle} shows two cases in which the morphism $g$ is $Z$-marked and $h$ is not.
The following example shows two $Z$-marked morphisms $g$ and $h$ which have two minimal solutions. 

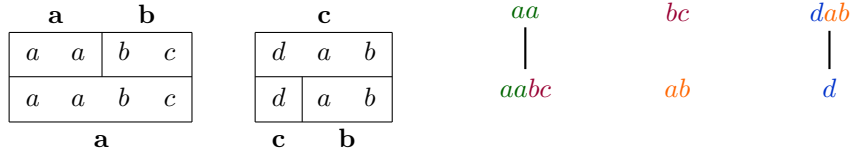
\begin{figure}[ht]
    \centering
    
    \begin{tikzpicture}[anchor=base, baseline, inner sep = 6pt] 
    \foreach \x/\y in {1/a,2/a,3/b,4/c}
    \node (dol\x) at (\x*6mm,0) {$\y$};
    \foreach \x/\y in {1/a,2/a,3/b,4/c}
    \node (hor\x) at (\x*6mm,.6) {$\y$};
    \draw (hor1.south west) -- (hor4.south east);
    \draw (dol1.south west) -- (dol4.south east);
    \draw (hor1.north west) -- (dol1.south west); 
    \draw (hor2.north west-|hor4.south east) -- (dol4.south east); 
    \draw (hor2.north west-|hor1.south west) -- (hor4.south east|-hor1.north west); 
    \draw (hor2.south east|-hor1.north west) -- (hor2.south east); 
    \node at (15mm,-.55) {$\a$};
    \node at (9mm,1.1) {$\a$};
    \node at (21mm,1.1) {$\b$};
    \end{tikzpicture}
    \quad\quad
    \begin{tikzpicture}[anchor=base, baseline, inner sep = 6pt] 
    \foreach \x/\y in {1/d,2/a,3/b}
    \node (dol\x) at (\x*6mm,0) {$\y$};
    \foreach \x/\y in {1/d,2/a,3/b}
    \node (hor\x) at (\x*6mm,.6) {$\y$};
    \draw (hor1.south west) -- (hor3.south east);
    \draw (dol1.south west) -- (dol3.south east);
    \draw (hor2.north west-|hor1.south west) -- (dol1.south west); 
    \draw (hor2.north west-|hor3.south east) -- (dol3.south east); 
    \draw (hor2.north west-|hor1.south west) -- (hor3.south east|-hor2.north west); 
	\draw (hor1.south east) -- (dol1.south east); 
    \node at (6mm,-.55) {$\c$};
    \node at (15mm,-.55) {$\b$};
    \node at (12mm,1.1) {$\c$};
   \end{tikzpicture}
    \quad\quad
    \def\Saa{\textcolor{mygreen}{aa}}	
	\def\Sbc{\textcolor{myred}{bc}} 
	\def\Sab{\textcolor{myorange}{ab}} 
	\def\Sd{\textcolor{myblue}{d}} 
	\begin{tikzpicture}
	\node (1a) at (0,1) {$\Saa$};
	\node (2a) at (2,1) {$\Sbc$};
	\node (3a) at (4,1) {$\Sd\Sab$};
	\node (1b) at (0,0) {$\Saa\Sbc$};
	\node (2b) at (2,0) {$\Sab$};
	\node (3b) at (4,0) {$\Sd$};
	\draw[thick] (1a)-- (1b);
	\draw[thick] (3a)-- (3b);
	\end{tikzpicture}	
    \caption{A representation of the solutions $(\a\b, \a), (\c,\c\b)$ and the free graph of two $Z$-marked morphisms from Example \ref{ex:twomarked}.}
    \end{figure}\label{fig:ghmarked}

\begin{example}\label{ex:twomarked}
Let
$$\begin{array}{ll}
g(\a)=aa   & h(\a)=aabc \\
g(\b)=bc    & h(\b)=ab \\
g(\c)=dab   & h(\c)=d\,.
\end{array}$$

The free basis of $Z$ is $B=\{aa, ab, bc, d\}$, $g$ and $h$ are both $Z$-marked, and the only two minimal solutions are $(\a\b, \a), (\c,\c\b)$.
In such a case each minimal solution introduces a vertical edge, which yields four connected components (cf. Figure \ref{fig:ghmarked}). 
\end{example}

\medskip
By symmetry, we shall suppose in what follows that $g$ is $Z$-marked, $\first_Z(h(\a))\neq\first_Z(h(\c))$ and $\first_Z(h(\b))\neq\first_Z(h(\c))$.

\medskip

Now we introduce the key ingredient of the proof of our theorem, namely  the definition of the critical overflow which was first introduced in \cite{Ehrenfeucht1983} (see also \cite[pp. 347--351]{Handbook}).
	We say that the word $o \in \Sigma^*$ is a \emph{critical overflow} if   
	$g(\u)= h(\vv)o$, for some $\u, \vv \in A^*$, and there are pairs $(\u_1,\u_2)$, $(\vv_1,\vv_2)$ in $A^* \times A^*$ such that $\first (\u_1) \neq \first (\u_2)$, $\first (\vv_1) \neq \first (\vv_2)$ and both $g(\u\u_1)=h(\vv\vv_1)$ and $g(\u\u_2)=h(\vv\vv_2)$. Moreover, we say that $o$ is {\em a critical overflow on} $(\u, \vv)$.

	Informally, if $o$ is a critical overflow on a pair $(\u, \vv)$, then $(\u, \vv)$ is a prefix of at least two distinct minimal solutions $(\u\u_1, \vv\vv_1)$ and $(\u\u_2, \vv\vv_2)$. It represents the situation when the continuation of $(\u, \vv)$ is not given uniquely during the construction of the minimal solution neither for $\u$ nor for $\vv$.

		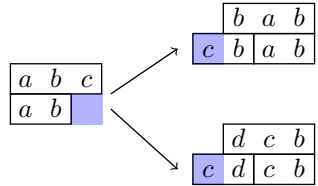
\begin{figure}[ht]
		\centering
 \begin{tikzpicture}[anchor=base,line width = .5pt]
\def\krok{4mm} 
\foreach \x/\y/\p in {0/3/o,1/1/a,1/2/b,1/3/c,0/1/a,0/2/b,2/7/c,2/8/b,2/9/a,2/10/b,3/8/b,3/9/a,3/10/b,{-1}/8/d,{-1}/9/a,{-1}/10/b,{-2}/7/c,{-2}/8/d,{-2}/9/a,{-2}/10/b}
{
	\node[minimum size = {\krok - .5pt}] (node\x\y) at (\y*\krok,\x*\krok) {};
}
\foreach \x/\y/\p in {1/1/a,1/2/b,1/3/c,0/1/a,0/2/b}
\foreach \x/\yaa/\ybb in {0/3/3,2/7/7,{-2}/7/7}
{
	\draw[color = blue!30, fill = blue!30] (node\x\yaa.south west) rectangle (node\x\ybb.north east);
} 
\foreach \x/\yaa/\ybb in {1/1/3,0/1/2,2/7/8,2/9/10,3/8/10,{-1}/8/10,{-2}/7/8,{-2}/9/10}
{
	\draw (node\x\yaa.south west) rectangle (node\x\ybb.north east);
}
\foreach \x/\y/\p in {1/1/a,1/2/b,1/3/c,0/1/a,0/2/b,2/7/c,2/8/b,2/9/a,2/10/b,3/8/b,3/9/a,3/10/b,{-1}/8/d,{-1}/9/c,{-1}/10/b,{-2}/7/c,{-2}/8/d,{-2}/9/c,{-2}/10/b}
{
	\node[anchor = base] at (\y*\krok,\x*\krok-3pt) {$\p$};
}
\draw[->] (3.8*\krok, .5*\krok) -- (6*\krok, 2*\krok);
\draw[->] (3.8*\krok, 0*\krok) -- (6*\krok,-2*\krok);
\end{tikzpicture}	
		\caption{A critical overflow of morphisms of Example \ref{ex:critical}.} \label{fig:critical}
	\end{figure}
	
	\begin{example}\label{ex:critical}
		Let $g(\a)=abc, g(\b)=bab$ and $g(\c)=dcb$, $h(\a)=ab, h(\b)=cb$ and $h(\c)=cd.$ Then $c$ is a critical overflow on $(\u,\vv)=(\a,\a)$ and two minimal solutions are $(\a\b,\a\b\a)$ and $(\a\c,\a\c\b)$. See Figure \ref{fig:critical} for a representation.     
	\end{example}{}	
	
	\begin{Remark}\label{rm:blocksoverflows}
Since $g$ and $h$ are morphism and $\FH Z$ is free, it follows that the critical overflows belongs to $\FH Z$.
\end{Remark}

The previous remark is a basic trivial property of free monoids and its free basis but it is fundamental for the proof of the following results that characterize the critical overflows in our setting and the corresponding properties of the free graph.
	
\begin{Remark}
	For sake of completeness, we should also consider the case when $o$ is nonempty and $g(\u)o= h(\vv)$ in the definition of the critical overflow. Note however, that such a situation is excluded by the hypothesis that $g$ is marked.
\end{Remark}{}	
	
	\begin{proposition}\label{prop:cases}
		If $(\r, \s)$ and $(\r', \s')$ are two distinct minimal solutions then there is a critical overflow $o$ on $(\u, \vv)$,  with $\u = \r \wedge \r'$ and $\vv = \s \wedge \s'$. Therefore, $h$ is Z-marked iff $o$ is an empty overflow.
	\end{proposition}{}
	\begin{proof}
		By lemma \ref{lm:notcompa}, the components of the two minimal solutions are not prefix comparable respectively. Therefore $\r = \u\u_1$, $\r' = \u \u_2$, $\s = \vv \vv_1$ and $\s' = \vv\vv_2$ where $\u = \r \wedge \r'$, $\vv = \s \wedge \s'$, and  all $\u_1$, $\u_2$, $\vv_1$, $\vv_2$ are nonempty. Let $\a=\first(\u_1)$, $\a'=\first(\u_2)$, $\b=\first(\vv_1)$ and $\b'=\first(\vv_2)$ where $\a\neq \a'$ and $\b\neq \b'$. The case $\u = \epsilon$ and $\vv \neq \epsilon$ is excluded by the assumption that $g$ is marked. We have the following cases:
		\begin{itemize}
			\item If $\u = \vv= \epsilon$, then the empty word is a critical overflow on $(\epsilon, \epsilon)$. Since $G_Z$ has two vertical edges $[g(\a), h(\b)]$ and $[g(\a'), h(\b')]$, it cannot have an horizontal edge, hence $h$ is Z-marked.
			\item If $\u \neq \epsilon$ and $\vv=\epsilon$, then $g(\u)$ is a nonempty critical overflow on $(\u, \epsilon)$. We have $h(\b)\neq h(\b')$, but $\first_Z(h(\b)) = \first_Z(h(\b'))$ i.e., $h$ is not $Z$-marked.
		
			\item Finally, if both $\u \neq \epsilon$ and $\vv \neq \epsilon$, then  we have $h(\vv) < g(\u)$ because $g$ is marked and moreover there is a nonempty critical overflow $o$ with $g(\u)=h(\vv)o$.  By Remark \ref{rm:blocksoverflows} we have that $\first_Z(o)=first_Z(h(\b))=first_Z(h(\b'))$, hence $h$ is not Z-marked.
		\end{itemize}{}
		\end{proof}{}

\begin{figure}[ht]
	\centering
	\begin{tikzpicture}[dot/.style={circle,fill=black,
		inner sep=1.5pt},scale=.5]
	\fill[gray!30, rounded corners = 5pt] (-.3,-.3) rectangle (2.3,1.3);
	\node[dot] (1a) at (0,1) {};
	\node[dot] (2a) at (1,1) {};
	\node[dot] (3a) at (2,1) {};
	\node[dot] (1b) at (0,0) {};
	\node[dot] (2b) at (1,0) {};
	\node[dot] (3b) at (2,0) {};
	\draw[thick] (1a)-- (1b);
	\draw[thick] (2a)-- (2b);
    \end{tikzpicture}	
\quad	
	\begin{tikzpicture}[dot/.style={circle,fill=black,
		inner sep=1.5pt},scale=.5]
	\fill[gray!30, rounded corners = 5pt] (-.3,-.3) rectangle (2.3,1.3);
	\node[dot] (1a) at (0,1) {};
	\node[dot] (2a) at (1,1) {};
	\node[dot] (3a) at (2,1) {};
	\node[dot] (1b) at (0,0) {};
	\node[dot] (2b) at (1,0) {};
	\node[dot] (3b) at (2,0) {};
	\draw[thick] (1a)-- (1b);
	\draw[thick] (1a)-- (2b);
	\draw[thick] (1b)-- (2b);
	\end{tikzpicture}
\quad	
    \begin{tikzpicture}[dot/.style={circle,fill=black,
		inner sep=1.5pt},scale=.5]
	\fill[gray!30, rounded corners = 5pt] (-.3,-.3) rectangle (2.3,1.3);
	\node[dot] (1a) at (0,1) {};
	\node[dot] (2a) at (1,1) {};
	\node[dot] (3a) at (2,1) {};
	\node[dot] (1b) at (0,0) {};
	\node[dot] (2b) at (1,0) {};
	\node[dot] (3b) at (2,0) {};
	\draw[thick] (1a)-- (1b);
	\draw[thick] (1b)-- (2b);
	\end{tikzpicture}	
\quad	
	\caption{Free graphs in the three cases of critical overflows.}\label{fig:freegraphcases}
\end{figure}
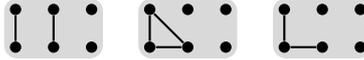

\begin{Remark}\label{rm:threecases}
   We can reformulate the three cases of critical overflows of the previous proof in terms of properties of $G_Z$ as follows.
	Let $(\r, \s)$ and $(\r', \s')$ be two distinct minimal solutions and $\u = \r \wedge \r'$, $\vv = \s \wedge \s'$. Let $\r = \u\u_1$, $\r' = \u \u_2$, $\s = \vv \vv_1$ and $\s' = \vv\vv_2$, $\a=\first(\r_1)$, $\a'=\first(\r_1')$, $\b=\first(\s_1)$ and $\b'=\first(\s_1')$ where $\a\neq \a'$ and $\b\neq \b'$. Then,
     \begin{enumerate}
	 \item \label{case:emptyov} If $\u = \vv= \epsilon$, $G_Z$ has two vertical edges $[g(\a), h(\b)]$ and $[g(\a'), h(\b')]$ (see the first case in Figure \ref{fig:freegraphcases}). Note that Example \ref{ex:twomarked} with Figure \ref{fig:ghmarked} show such a situation.
	 \item \label{case:triangle} If $\u \neq \epsilon$ and $\vv=\epsilon$, $G_Z$ has two vertical edges $[g(\first(\u)), h(\b)]$ and $[g(\first(\u)), h(\b')]$ and an horizontal edge $[h(\b), h(\b')]$  creating a connected component of three nodes (see the second case in Figure \ref{fig:freegraphcases}). Example \ref{ex:triangle} with Figure \ref{fig:triangle} verify such a case.
	 \item  \label{case:overflow} Finally, if both $\u \neq \epsilon$ and $\vv \neq \epsilon$, $G_Z$ has a vertical edge $[g(\first(\u)), h(\first(\vv))]$ and a horizontal edge $[h(\b), h(\b')]$ (see the first case in Figure \ref{fig:freegraphcases}).  Note that Example \ref{ex:infinit} and Figure \ref{fig:infinit} show such a situation.
	\end{enumerate}{}
	Note that in any of this cases $G_Z$ cannot have further edges. 
\end{Remark}

	\begin{lemma} \label{lm:criticaloverflow}
		Let $o$ be a critical nonempty overflow such that $g(\u)=h(\vv)o$ with $\u,\vv \in A^*$. 
		Let $\u_1,\u_2, \vv_1, \vv_2 \in A^+$ be such that $g(\u\u_1)=h(\vv\vv_1)$ and $g(\u\u_2)=h(\vv\vv_2)$ and \[\a=\first(\u_1)\neq \first(\u_2) = \a',\quad \b=\first(\vv_1)\neq\first(\vv_2) = \b'\,.\]  Then 
		$$o = h(\b) \wedge_Z h(\b').$$ 
	\end{lemma}
	\begin{proof}
		Note that $o$ is prefix comparable with both $h(\b)$ and $h(\b')$. If $h(\b) \leq o$ or $h(\b') \leq o$, then also $h(\b)$ and $h(\b')$ are prefix comparable, contradicting Proposition \ref{prop:bifix}. 
		Moreover $o \leq h(\b) \wedge_Z h(\b')$, by Remark \ref{rm:blocksoverflows}.

		Let $o<h(\b) \wedge_Z h(\b')$ and let $o =(h(\b) \wedge_Z h(\b'))o'$. Then $\first_Z(g(\a))=\first_Z(g(\a')) = \first_Z(o')$, a contradiction with $g$ being Z-marked. Then $o = h(\b) \wedge_Z h(\b')$. 
	\end{proof}{}

By Lemma \ref{lm:marked} and Lemma \ref{lm:criticaloverflow} we have the uniqueness of the critical overflow.

 We can now prove the main result of the paper.

 \begin{thm}
 Let $X=\{x, y, z\}^*$ , $U=\{u, v, w\}^*$ be different $3$-maximal submonoids of $\Sigma^*$. Then\\  
$X^* \cap U^*=\{\alpha, \beta\}^*$, for some $\alpha, \beta\in \Sigma^*$ \\or \\
$X^* \cap U^*=\{\alpha \gamma^* \beta\}^*$, for some $\alpha, \beta, \gamma \in \Sigma^+.$
 \noindent 
 \end{thm}
 \begin{proof}
Let $(\r_1,\s_1)$, $(\r_2, \s_2)$ and $(\r_3, \s_3)$ be three minimal solutions. If two of them, say $(\r_1, \s_1)$ and $(\r_2, \s_2)$ are such that $\r_1 \wedge \r_2=\epsilon$ then, by Remark \ref{rm:threecases} case \ref{case:emptyov}, $G_Z$ has two vertical edges and cannot have others edges. Hence, if $\r_1 \wedge \r_3= \epsilon$ and $\r_2 \wedge \r_3= \epsilon$ then, by Remark \ref{rm:threecases} case \ref{case:emptyov}, $G_Z$ must have another vertical edge hence we have a contradiction. If $\r_1 \wedge \r_3= \epsilon$ and $\r_2 \wedge \r_3 \neq \epsilon$ (resp. $\r_1 \wedge \r_3 \neq \epsilon$ and $\r_2 \wedge \r_3 = \epsilon$), then, by Remark \ref{rm:threecases} case \ref{case:overflow}, there is a nonempty critical overflow i.e. $h$ is not Z-marked and an horizontal edge exists, again a contradiction.  

It follows that $\r_1 \wedge \r_2 \wedge \r_3 \neq \epsilon$. Let  
$$\u=\r_1 \wedge \r_2 \,\,\,\mbox{and} \,\,\, \vv=\s_1 \wedge \s_2$$
and
$$\u'=\r_2 \wedge \r_3 \,\,\,\mbox{and} \,\,\, \vv'=\s_2 \wedge \s_3.$$ 

By Proposition \ref{prop:cases}, we have
\begin{equation}\label{eq:overflow}
 g(\u)=h(\vv)o\,\, \mbox{and} \,\, g(\u')=h(\vv')o.
\end{equation}
where $o$ is the (unique) critical overflow.

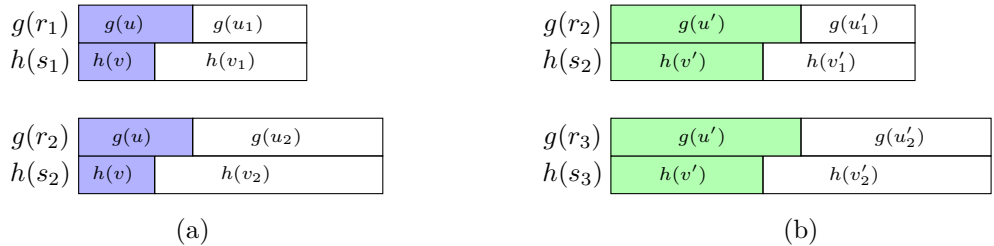
\begin{figure}[ht]
	   \centering
  \begin{tikzpicture}
    \node at (0,0) {$h(s_2)$};
  	\draw[fill = blue!30] (.5,-.25) rectangle (1.5,.25);
  	\node at (0.95,0) {${\scriptstyle h(v)}$};
  	\draw (1.5,-.25) rectangle (4.5,.25);
  	\node at (2.7,0) {${\scriptstyle h(v_2)}$};
    \node at (0,.5) {$g(r_2)$};
  	\draw[fill = blue!30] (.5,.25) rectangle (2,.75);
  	\node at (1.2,.5) {${\scriptstyle g(u)}$};
  	\draw (2,.25) rectangle (4.5,.75);
  	\node at (3.1,.5) {${\scriptstyle g(u_2)}$};
\begin{scope}[shift = {(7,1.5)}]
       \node at (0,0) {$h(s_2)$};
   \draw[fill = green!30] (.5,-.25) rectangle (2.5,.25);
   \node at (1.45,0) {${\scriptstyle h(v')}$};
   \draw (2.5,-.25) rectangle (4.5,.25);
   \node at (3.4,0) {${\scriptstyle h(v_1')}$};
   \node at (0,.5) {$g(r_2)$};
   \draw[fill = green!30] (.5,.25) rectangle (3,.75);
   \node at (1.7,.5) {${\scriptstyle g(u')}$};
   \draw (3,.25) rectangle (4.5,.75);
   \node at (3.7,.5) {${\scriptstyle g(u_1')}$};
\end{scope}
\begin{scope}[shift = {(0,1.5)}]
    \node at (0,0) {$h(s_1)$};
\draw[fill = blue!30] (.5,-.25) rectangle (1.5,.25);
\node at (0.95,0) {${\scriptstyle h(v)}$};
\draw (1.5,-.25) rectangle (3.5,.25);
\node at (2.5,0) {${\scriptstyle h(v_1)}$};
\node at (0,.5) {$g(r_1)$};
\draw[fill = blue!30] (.5,.25) rectangle (2,.75);
\node at (1.1,.5) {${\scriptstyle g(u)}$};
\draw (2,.25) rectangle (3.5,.75);
\node at (2.6,.5) {${\scriptstyle g(u_1)}$};
\end{scope}
\begin{scope}[shift = {(7,0)}]
\node at (0,0) {$h(s_3)$};
\draw[fill = green!30] (.5,-.25) rectangle (2.5,.25);
\node at (1.45,0) {${\scriptstyle h(v')}$};
\draw (2.5,-.25) rectangle (5.5,.25);
\node at (3.7,0) {${\scriptstyle h(v_2')}$};
\node at (0,.5) {$g(r_3)$};
\draw[fill = green!30] (.5,.25) rectangle (3,.75);
\node at (1.7,.5) {${\scriptstyle g(u')}$};
\draw (3,.25) rectangle (5.5,.75);
\node at (4.3,.5) {${\scriptstyle g(u_2')}$};
\end{scope}
\node at (2,-.75) {(a)};
\node at (10,-.75) {(b)};
\end{tikzpicture}
    \caption{The representation of three minimal solutions.}\label{fig:r1r2r3}
\end{figure}

Moreover, by Lemma \ref{lm:notcompa} there exist $\u_1 \neq \u_2 \neq \epsilon, \vv_1 \neq \vv_2 \neq \epsilon$ with 
\begin{align*}
\a_1&=\first(\u_1) \neq \first (\u_2)=\a_2, & \b_1&=\first(\vv_1) \neq \first (\vv_2)=\b_2 
\end{align*}
such that $$(\r_1, \s_1)=(\u\u_1,\vv\vv_1),\,\, (\r_2, \s_2)=(\u\u_2,\vv\vv_2)$$ 
(cf. Figure \ref{fig:r1r2r3}(a)),  
and there exist $\u_1' \neq \u_2'\neq \epsilon$,  $\vv_1'\neq  \vv_2'\neq \epsilon$  with
\begin{align*}
 \a'_1 &=\first(\u_1') \neq \first (\u_2')=\a_2', &  \b_1'&=\first(\vv_1') \neq \first (\vv_2')=\b_2'
 \end{align*}
 such that $$(\r_2, \s_2)=(\u'\u_1',\vv'\vv_1'),\,\, (\r_3, \s_3)=(\u'\u_2',\vv'\vv_2')$$ 
(cf. Figure \ref{fig:r1r2r3}(b)). 

First, we prove that $\u \neq \u'$. Indeed, if $\u=\u'$ then, by (\ref{eq:overflow}), we have  $\vv=\vv'$, with $\b_1 \neq   \b_2 \neq \b_2'$. By Lemma \ref{lm:criticaloverflow}, we have $o=h(\b_1) \wedge_Z h(\b_2)=h(\b_2) \wedge_Z h(\b_2')$, i.e., we have three different elements of $h(A)$ having a nonempty common prefix. Then $G_Z$ has two distinct horizontal edges, a contradiction.

Suppose (without loss of generality) that $\u < \u'$. From (\ref{eq:overflow}) and Lemma \ref{lm:incmorphism}, it follows that $\vv < \vv'$. Let $\u'=\u\p$, and $\vv'=\vv\q$. Then $\p,\q \neq \epsilon$ and $og(\p)=h(\q)o$.

 Canceling, if necessary, superfluous factors $(\w,\w')$ satisfying $og(\w)=h(\w')o$, one can choose $(\r_1,\s_1)$, $(\r_2, \s_2)$ and $(\r_3, \s_3)$ with the following properties:
\begin{enumerate}\label{pos}
	\item \label{pos:uv} $(\u, \vv)$ is such that for any $(\overline{\u}, \overline{\vv})<(\u, \vv)$, $g(\overline{\u}) \neq h(\overline{\vv})o$. 
	\item \label{pos:upvq} $(\u', \vv')$ is such that for any $(\u, \vv) < (\overline{\u}, \overline{\vv})<(\u\p, \vv\q)$, $g(\overline{\u}) \neq h(\overline{\vv})o$. 
	\item \label{pos:u1v1} $(\u_1, \vv_1)$ is such that for any $(\u, \vv) < (\overline{\u}, \overline{\vv})<(\u\u_1, \vv\vv_1)$, $g(\overline{\u}) \neq h(\overline{\vv})o$. 
\end{enumerate} 

We prove that the set $\{(\u\p^i\u_1, \vv\q^i\vv_1)\mid i \geq 0\}$ is the set of all the minimal solutions, i.e., we prove that a pair $(\r, \s)$ is a minimal solution iff $(\r, \s)=(\u\p^i\u_1, \vv\q^i\vv_1)$ for a certain $i \geq 0$. 

If  $(\r, \s) \neq (\u\u_1, \vv\vv_1)$ and $(\r, \s) \neq (\u\p\u_1, \vv\q\vv_1)$ we have that $\r \wedge \u\u_1 \neq \epsilon$ and $\r \wedge \u\p\u_1 \neq \epsilon$. Indeed we saw that for any three minimal solutions the first components must have a nonempty common prefix.

\begin{figure}[ht]
	   \centering
\begin{tikzpicture}
\draw[draw = yellow!80, fill = yellow!80] (1.5,-.25) rectangle (2,.25);
\node at (1.75,0) {${\scriptstyle o}$};
\draw (.5,-.25) rectangle (1.5,.25);
\node at (0.95,0) {${\scriptstyle h(v)}$};
\draw (1.5,-.25) rectangle (4.5,.25);
\node at (2.7,0) {${\scriptstyle h(v_1)}$};
\draw(.5,.25) rectangle (2,.75);
\node at (1.2,.5) {${\scriptstyle g(u)}$};
\draw (2,.25) rectangle (4.5,.75);
\node at (3.1,.5) {${\scriptstyle g(u_1)}$};
\begin{scope}[shift = {(0,-1.5)}]
\begin{scope}[shift = {(1,0)}]
\draw[draw = yellow!80, fill = yellow!80] (1.5,-.25) rectangle (2,.25);
\node at (1.75,0) {${\scriptstyle o}$};
\draw (1.5,-.25) rectangle (4.5,.25);
\node at (2.7,0) {${\scriptstyle h(v_1)}$};
\draw (2,.25) rectangle (4.5,.75);
\node at (3.1,.5) {${\scriptstyle g(u_1)}$};
\end{scope}
\draw[draw = yellow!80, fill = yellow!80] (1.5,-.25) rectangle (2,.25);
\draw (1.5,-.25) rectangle (2.5,.25);
\node at (2.1,0) {${\scriptstyle h(q)}$};
\draw (2,.25) rectangle (3,.75);
\node at (2.5,.5) {${\scriptstyle g(p)}$};
\draw (.5,-.25) rectangle (1.5,.25);
\node at (0.95,0) {${\scriptstyle h(v)}$};
\draw(.5,.25) rectangle (2,.75);
\node at (1.2,.5) {${\scriptstyle g(u)}$};
\end{scope}
\begin{scope}[shift = {(0,-3)}]
\begin{scope}[shift = {(2,0)}]
\draw[draw = yellow!80, fill = yellow!80] (1.5,-.25) rectangle (2,.25);
\node at (1.75,0) {${\scriptstyle o}$};
\draw (1.5,-.25) rectangle (4.5,.25);
\node at (2.7,0) {${\scriptstyle h(v_1)}$};
\draw (2,.25) rectangle (4.5,.75);
\node at (3.1,.5) {${\scriptstyle g(u_1)}$};
\end{scope}
\draw[draw = yellow!80, fill = yellow!80] (1.5,-.25) rectangle (2,.25);
\draw (1.5,-.25) rectangle (2.5,.25);
\node at (2.1,0) {${\scriptstyle h(q)}$};
\draw (2,.25) rectangle (3,.75);
\node at (2.5,.5) {${\scriptstyle g(p)}$};
\begin{scope}[shift={(1,0)}]
\draw[draw = yellow!80, fill = yellow!80] (1.5,-.25) rectangle (2,.25);
\draw (1.5,-.25) rectangle (2.5,.25);
\node at (2.1,0) {${\scriptstyle h(q)}$};
\draw (2,.25) rectangle (3,.75);
\node at (2.5,.5) {${\scriptstyle g(p)}$};
\end{scope}
\draw (.5,-.25) rectangle (1.5,.25);
\node at (0.95,0) {${\scriptstyle h(v)}$};
\draw(.5,.25) rectangle (2,.75);
\node at (1.2,.5) {${\scriptstyle g(u)}$};
\end{scope}
\end{tikzpicture}
    \caption{Three minimal solutions $(\u\p^i\u_1, \vv\q^i\vv_1)$, with $i=0,1,2$.}\label{fig:cicle}
\end{figure}{}

 If $\r \wedge \u\u_1= \overline{\u} < \u$, then, by Proposition \ref{prop:cases} and Lemma \ref{lm:criticaloverflow}, we have $g(\overline{\u})=h(\overline{\vv})o$, where $\overline{\vv}=\s \wedge \vv\vv_1$. By \eqref{eq:overflow}, we have $\overline{\vv} < \vv$ which is a contradiction with the assumption \ref{pos:uv}.

 If $\r \wedge \u\p\u_1=\u\p_1$, with $\p_1\ < \p$, let $\overline{\vv}=\s \wedge \vv\q\vv_1$, then by Proposition \ref{prop:cases} and Lemma \ref{lm:criticaloverflow}, we have $g(\overline{\u})=h(\overline{\vv})o$.  By \eqref{eq:overflow} $\overline{\vv}=\vv\q_1$, with $\q_1 < \q$, against assumption \ref{pos:upvq}.

 If there exists $k \in \{0,1\}$ such that $\overline{\u}=\r \wedge \u\p^k\u_1=\u\p^k\t_1$, with $\t_1< \u_1$, then, by Proposition \ref{prop:cases} and Lemma \ref{lm:criticaloverflow}, we have $g(\overline{\u})=h(\overline{\vv})o$ where    $\overline{\vv}=\vv\q^k\w_1$ with $\w_1 < \vv_1$ from (\ref{eq:overflow}). Since $g(\u\p^k\t_1)=h(\vv\q^k\w_1)o$ it follows that $g(\u\t_1)=h(\vv\w_1)o$ against assumption \ref{pos:u1v1}. 

We can conclude that $(\r,\s)=(\u\p\t, \vv\q\w)$, where $\t, \w \neq \epsilon$, and $(\u\t, \vv\w)$ is a minimal solution.
By induction, it follows that $(\t,\w)=(\p^{i}\u_1,\q^{i}\vv_1)$, for some $i>0$.

Finally, $(\u\p^{i}\u_1,\vv\q^{i}\vv_1)$, is minimal for each $i \geq 0$. Indeed, if for some $i$ there exists a minimal solution $(\r,\s)<(\u\p^{i}\u_1,\vv\q^{i}\vv_1)$, as we have seen, $(\r,\s)=(\u\p^{j}\u_1,\vv\q^{j}\vv_1)$ for some $j\geq 0$. It follows that either $\u_1 < \p$ or $\p < \u_1$. In the first case we have the contradicion $\u\u_1 < \u\p\u_1$, i.e. $\r_1$ and $\r_2$ are prefix comparable, in the second case we contradict $\u= \r_1 \wedge \r_2$.  
The thesis follows with $\alpha=g(\u), \gamma=g(\p)$ and $\beta=g(\u_1)$. 
\end{proof}

\section{Conclusions}
The hypothesis of $k$-maximality considerably simplifies the structure of the intersection of monoids and gives an interesting connection with the generation of binary equality set, in the case $k=3$. Our proof also shows the importance of the free graph in this context. Together with combinatorial properties of $k$-maximal monoids investigated in \cite{Castiglione_2019}, this is promising for further investigation of cases with arbitrary $k$.

\bibliographystyle{plain}
\bibliography{ref}

\end{document}